\documentclass[runningheads, envcountsame, a4paper]{llncs}

\usepackage{\jobname}

\begin{document}

\title{The Weight in Enumeration}
\toctitle{The Weight in Enumeration}

\author{Johannes~Schmidt}
\tocauthor{Johannes~Schmidt}

\institute{J{\"o}nk{\"o}ping International Business School, J{\"o}nk{\"o}ping University, Sweden\\
  \texttt{johannes.schmidt@ju.se}}
\maketitle

\begin{abstract}
In our setting enumeration amounts to generate all solutions of a problem instance without duplicates.
We address the problem of enumerating the models of $B$-formul{\ae}.
A $B$-formula is a propositional formula whose connectives are taken from a fixed set $B$ of Boolean connectives.
Without imposing any specific order to output the solutions, this task is solved.
We completely classify the complexity of this enumeration task for all possible sets of connectives $B$ imposing the orders of (1) non-decreasing weight, (2) non-increasing weight; the weight of a model being the number of variables assigned to 1.
We consider also the weighted variants where a non-negative integer weight is assigned to each variable and show that this add-on leads to more sophisticated enumeration algorithms and even renders previously tractable cases intractable, contrarily to the constraint setting.
As a by-product we obtain also complexity classifications for the optimization problems known as $\MINONES$ and $\MAXONES$
which are in the $B$-formula setting two different tasks.
\keywords{Computational Complexity, Enumeration, non-decreasing weight, Polynomial delay, Post's Lattice, MaxOnes}
\end{abstract}

\section{Introduction}

We deal in this paper with algorithmic and complexity of \emph{enumeration}, the task of generating all solutions of a problem instance.
Over the last 15 years, in both practice and theory, one can observe a growing interest in studying enumeration problems which have previously been poorly studied compared to decision, optimization and counting problems.
The main reason for this may lie in the huge increase of the size of the data computers are nowadays demanded and able to process in everyday applications.

It is in the meanwhile commonly agreed to consider an enumeration algorithm \emph{efficient} if it has \emph{polynomial delay} (\cite{JoPaYa88,Schmidt09}), \ie, the time passing between outputs of two successive solutions is polynomial in the input size (while the total time of the output process is usually exponential, due to large solution sets).
Variants and different degrees of efficiency in this context exist, see e.g. \cite{JoPaYa88,StrozeckiPhD10}.
Known reductions for enumeration are essentially one-to-one parsimonious reductions, as opposed to counting complexity where a greater variety of useful reductions exist, see e.g. \cite{DuHeKo05,BuDyGoJaJeRe12}.
An interesting issue of enumeration is the order in which the solutions are output. Imposing different orders for an enumeration process may drastically change the complexity, see e.g. \cite{JoPaYa88,CrOlSc11,BoCrGaReScVo12}.

We focus in this paper on the task of enumerating the models of a propositional formula. This task has already been addressed in the context of Boolean constraint satisfaction problems (CSPs). One considers here formul{\ae} in generalized conjunctive normal form \cite{Schaefer78}, also called \emph{$\Gamma$-formul{\ae}} where $\Gamma$ is the \emph{constraint language}. In \cite{CrHe97} this task, $\EnumSAT$ for short, has been studied without imposing any special order. There is a polynomial delay algorithm if and only if the underlying constraint language $\Gamma$ is either Horn, or dual Horn, or affine, or 2CNF, unless $\P = \NP$.
It is worth mentioning that the algorithms underlying this result are all straight forward extensions of the corresponding decision procedures via the notion of \emph{self-reducibility} \cite{Schnorr76} which naturally leads to lexicographic order. In the non-Boolean domain the self-reducible fragment does not deliver all tractable cases anymore \cite{ScSc07} and things get much more involved.


Back to the Boolean domain, $\EnumSAT$ has also been considered imposing the order of non-decreasing weight ($\EnumSATinc$ for short), the weight of a model being the number of variables assigned to 1.
The weight is a natural parameter in Boolean CSPs that can be assimilated to the cost of an assignment. Hence, the task $\EnumSATinc$ can be seen as the task of enumerating the cheapest solutions first, then the more expensive ones in order of increasing cost.
In \cite{CrOlSc11} the task $\EnumSATinc$ has been studied for $\Gamma$-formul\ae.
There is a polynomial delay algorithm to enumerate the models of a propositional $\Gamma$-formula by order of non-decreasing weight if and only if $\Gamma$ is width-2-affine or Horn, unless $\P = \NP$. By duality in that context, the task of enumerating by order of non-increasing weight, $\EnumSATdec$ for short, is tractable if and only if $\Gamma$ is width-2-affine or dual Horn.

In this paper we reveal new tractable fragments of propositional logic for $\EnumSATinc$ and $\EnumSATdec$ by considering fragments of propositional logic by a different approach.
A \emph{$B$-formula} is a propositional formula whose connectives are taken from $B$, a fixed set of Boolean functions.
This approach covers different fragments than the classical constraint approach, e.g. monotonic, self-dual, 0-separating of degree $n$.
It has first been taken by Lewis \cite{Lewis79} who showed that the satisfiability problem for $B$-formul{\ae}, $\SAT(B)$ for short, is $\NP$-complete if and only if the set $B$ is able to express negation of implication ($x \land \neg y$), unless $\P = \NP$. Since then, a number of problems dealing with propositional formul{\ae} have been parameterized by $B$-formul{\ae} in order to get a finer classification of their complexity, e.g. equivalence \cite{Reith03}, implication \cite{BeMeThVo08imp}, circumscription \cite{Thomas09}, abduction \cite{CrScTh12}.

In \cite{BoCrGaReScVo12} the model enumeration problem has been studied in the context of $B$-circuits without imposing an order and imposing lexicographic order. Roughly speaking, a $B$-formula can be represented by a $B$-circuit without size-increase, but, in general, not vice versa. Therefore, tractability translates from $B$-circuits to $B$-formul{\ae}, whereas this does not automatically hold for hardness results. We observe however that only slight modifications in the hardness proof from \cite{BoCrGaReScVo12} suffice and one obtains the same classification for $B$-formul{\ae}.

Our main contribution lies in complete classifications for $\EnumSATinc$ and $\EnumSATdec$. We show that the models of a $B$-formula can efficiently be enumerated  by order of non-decreasing weight if and only if the connectives are either 0-separating, or affine, or conjunctive, or disjunctive, unless $\P = \NP$.
%
We further show that we can efficiently enumerate by order of non-increasing weight if and only if the connectives are either 0-separating of degree 2, or monotone, or affine, unless $\P = \NP$.
We also consider the weighted variants of $\EnumSATinc$ and $\EnumSATdec$
(denoted $\wEnumSATinc$ and $\wEnumSATdec$, respectively) where a
weight function $w: \{x_1, \dots, x_n\} \rightarrow \N$ assigns a non-negative integer weight to each variable and the weight of an assignment is the sum of the weights of the variables assigned to $1$. We show that for $\wEnumSATinc$ the previously tractable fragment of 0-separating connectives now compounds intractable cases.

We also shed new light on the optimization problems known as $\MINONES$ and $\MAXONES$ where the task is to find a model of minimal / maximal weight.
We use these tasks, together with their weighted and non-trivial variants, to obtain hardness of the enumeration problems.
These two tasks are in our setting not "the same": contrary to the classical constraint setting \cite{KhSuWi97}, no duality notion allows to easily derive the classification for $\MAXONES$ from the one for $\MINONES$, or vice versa. This is because the duality notion in our setting transforms $\MINONES$ (find a satisfying assignment with minimal number of 1's) into the task of finding a non-satisfying assignment with maximal number of 1's.
We show further that allowing weights on the variables renders previously tractable fragments intractable, contrarily to the classical constraint approach.


Among the algorithmic enumeration strategies we use, we apply a method we shall call \emph{priority queue method}. It has first been used in \cite{JoPaYa88} in order to enumerate all
maximal independent sets of a graph in lexicographical order. This method turned out to be applicable in much more generality \cite{KiSa06,Schmidt09,CrOlSc11}.
We use it to obtain various polynomial delay algorithms for $\EnumSATinc$ and $\EnumSATdec$ and their weighted variants.

We give another non-trivial enumeration algorithm for $\EnumSATdec$ for the fragment of connectives that are 0-separating of degree 2 (Proposition~\ref{prop:EnumSATdec-S02-inDelayP})
that may be intuitively best described by \emph{nested} or \emph{incremental bruteforce}: we use the Erd\H{o}s-Ko-Rado Theorem \cite{ErKoRa61} to obtain a combinatorial bound that allows us to \emph{buy time} \cite{ScSc07} from a relatively large number of models whose output process delivers then enough time to compute further, computationally more involving models that are stored and output afterwards.

The paper is organized as follows. In Section 2 we give the necessary preliminaries on complexity theory, propositional formul{\ae} and clones of Boolean functions. In Section 3 we briefly look at model enumeration without order prescription.
We treat the order of non-decreasing and non-increasing weight in Sections 4 and 5 respectively. We conclude in Section 6.
\section{Preliminaries}\label{sec:preliminaries}
\subsection{Complexity Theory}

For the decision problems the arising complexity degrees encompass the classes $\P$ and $\NP$. For our hardness results we employ logspace many-one reductions.

An \emph{enumeration problem} $E$ can be formalized by a triple $(I, Sol, \leq)$,
where $I$ are the instances, $Sol$ is a function mapping each instance $x\in I$ to its set of solutions $Sol(x)$ and $\leq$ is a partial order (possibly empty) on the solution space.
We say that an algorithm $A$ solves an enumeration problem $E=(I, Sol, \leq)$ if
for a given input $x\in I$, $A$ generates one by one the elements of $Sol(x)$ without repetition such that
for all $y,z \in Sol(x)$ such that $y < z$, $A$ outputs $y$ before $z$.

An enumeration algorithm runs in \emph{polynomial delay} if the delay until the first solution is output and thereafter the delay between any two consecutive solutions is bounded by a polynomial $p(n)$ in the input size $n$.
We denote \emph{DelayP} the class of enumeration problems that admit a polynomial delay algorithm and \emph{SpaceDelayP} those problems in $\DelayP$ that are solvable within polynomial space.


%
%

%

\subsection{Propositional Formul{\ae}}

We assume familiarity with propositional logic. 
For a propositional formula $\varphi$ we denote by $\Vars{\varphi}$ the set of variables occurring in $\varphi$.
We represent an assignment $\sigma: \Vars{\varphi} \to \{0,1\}^n$ usually as a tuple over $\{0,1\}$ or when convenient by the set of variables assigned to $1$, \ie, the empty set corresponds to $\vec{0}$ and $\Vars{\varphi}$ to $\vec{1}$.
A \emph{model} for a formula $\varphi$ is an assignment that satisfies $\varphi$.
A \emph{non-trivial} assignment is an assignment different from $\vec{0}$ and $\vec{1}$.
The \emph{complement} of an assignment $\sigma$ is defined as $\overline{\sigma}(x) = 0 \Leftrightarrow \sigma(x) = 1$.
We call a variable $x \in \Vars{\varphi}$ \emph{fictive}, if the assignment $x = 0$ can be extended to a model of $\varphi$ if and only if so can the assignment $x = 1$.
We denote by $\varphi[\alpha/\beta]$ the formula obtained from $\varphi$ by replacing all occurrences of $\alpha$ with $\beta$.

\subsection{Clones of Boolean Functions}\label{sec:clones}

	A \emph{Boolean function} is an $n$-ary function $f: \{0,1\}^n \rightarrow \{0,1\}$.
	For technical reasons we consider only Boolean functions of arity $> 0$.
	It is not difficult but just technical to include also functions of arity $0$ into
	our considerations.
	We denote the $n$-ary Boolean constants by $\false^n$ and $\true^n$, respectively.
	When the arity is not relevant, we indicate them also by $\false$ and $\true$,
	keeping in mind that they have at least one fictive coordinate.
	An $n$-ary assignment $m$ such that $f(m) = 1$ will be called \emph{model} of $f$.
  A \emph{clone} is a set of Boolean functions that is closed under superposition,
  \ie, it contains all projections 
  (that is, the functions $f(a_1, \dots , a_n) = a_k$ for all $1 \leq k \leq n$ and $n \in \N$) 
  and is closed under arbitrary composition.
  Let $B$ be a finite set of Boolean functions. 
  We denote by $[B]$ the smallest clone containing $B$ and call $B$ a \emph{base} for $[B]$.
  In 1941 Post identified the set of all clones of Boolean functions \cite{pos41}.
  He gave a finite base for each of the clones and showed that they form a lattice under the usual $\subseteq$-relation, 
  hence the name \emph{Post's lattice} (see, \emph{e.g.}, Figure~\ref{fig:complexity}).
  To define the clones we introduce the following notions, where $f$ is an $n$-ary Boolean function:
\begin{itemize}
  \itemsep 5pt 
  \item $f$ is \emph{$c$-reproducing} if $f(c, \ldots , c) = c$,
    $c \in \{0,1\}$.
%
  \item $f$ is \emph{monotonic} if $a_1 \leq b_1, \ldots , a_n \leq b_n$
  	implies $f(a_1, \ldots , a_n) \leq f(b_1, \ldots , b_n)$.
%
  \item $f$ is \emph{$c$-separating of degree $k$} if 
    for all $A \subseteq f^{-1}(c)$ of size $|A|=k$ 
    there exists an $i \in \{1, \ldots , n\}$ 
    such that $(a_1, \ldots , a_n) \in A$ implies $a_i = c$,
    $c \in \{0,1\}$.
%
    \item $f$ is \emph{$c$-separating} if 
      $f$ is $c$-separating of degree $|f^{-1}(c)|$.
%
    \item $f$ is \emph{self-dual} if $f \equiv \mathrm{dual}(f)$, where
    	$\mathrm{dual}(f)(x_1,\ldots,x_n) := \neg f(\neg x_1, \ldots , \neg x_n)$.
%
    \item $f$ is \emph{affine} if $f \equiv x_1 \xor \cdots \xor x_n \xor c$ with
    	$c \in \{0,1\}$.
  \end{itemize}

	\noindent
  A list of some clones with definitions and finite bases is given in Table~\ref{tab:clones}.
  
	We will often add some function $f\notin C$ to a clone $C$
	and consider the clone $C' = [C \cup \{f\}]$ generated out of $C$ and $f$. With Post's lattice one can determine this $C'$ quite easily: It is the
	lowest clone above $C$ that contains $f$. We will use in particular the identities 
  $[\CloneS_{12} \cup \{\true\}] = \CloneS_1$, $[\CloneD \cup \{\true\}] = \CloneBF$, $[\CloneR_1 \cup \{\false\}] = \CloneBF$, $[\CloneD_1 \cup \{\true\}] = \CloneR_1$, and $[\CloneS_{10} \cup \{\true\}] = \CloneM_1$.
	
  \medskip
  
  A propositional formula using only connectives from $B$ is called a \emph{$B$-formula}.
  
\begin{definition}
	Let $f$ be an $n$-ary Boolean function and let $B$ be a set of Boolean functions.
	A $B$-formula $\varphi$ with $\Vars{\varphi} = \{x_1, \dots, x_k\}$ is called $B$-\emph{representation} of $f$ if
	there is an index function $\pi: \{1, \dots, n\} \rightarrow \{1, \dots, k\}$ such that $\forall x_1, \dots, x_k \in \{0,1\}$ it holds
$f(x_{\pi(1)}, \dots, x_{\pi(n)}) = 1$ if and only if $\varphi \text{ evaluates to }1$.
\end{definition}
We note that such a $B$-representation exists for every $f \in [B]$. We note further that, if $f$ does not contain fictive coordinates, then there is also a $B$-representation for $f$ without fictive variables. We shall keep this in mind, since some problems we consider are not stable under introduction/elimination of fictive variables.
	
There is a canonical transformation of a $B_1$-formula $\varphi_1$
into a $B$-formula, if $B_1 \subseteq [B]$: replace every connective in $\varphi_1$ by its $B$-representation. Though, this may lead to an explosion of the formula size. This can happen when a $B$-representation for some $f \in [B]$ uses some input variable more than once and $\varphi_1$ is of linear nesting depth, see e.g. \cite{CrScTh12}.
We will nevertheless use this transformation idea in order to obtain reductions. This is possible since in the cases we encounter, we are always able to (re-)write $\varphi_1$ as a formula of logarithmic nesting depth. We note that this is not possible in general.
We call formul{\ae} of logarithmic nesting depth \emph{compact}.
  


  \begin{table*}
    \centering
    \fontsize{8pt}{10.4pt}\selectfont
    \begin{tabular}{p{1cm}ll}
      \specialrule{\heavyrulewidth}{0cm}{0cm}
      Name & Definition & Base \\
      \specialrule{\heavyrulewidth}{0cm}{0.03cm}
      $\CloneBF$ & All Boolean functions & $\{x \land y, \neg x\}$ \\
      \hline
      $\CloneR_1$ & $\{f \mid f \text{ is $1$-reproducing}\}$ & $\{x \lor y, x = y\}$ \\
      \hline
      $\CloneR_2$ & $\CloneR_0 \cap \CloneR_1$ & $\{\lor, x \land (y = z) \}$ \\
      \hline
      $\CloneM$ & $\{f \mid f \text{ is monotonic}\}$ & $\{x \lor y, x \land y, \false, \true\}$ \\
      \hline

      $\CloneS^n_{0}$ & $\{f \mid f \text{ is $0$-separating of degree } n\}$ & $\{x \to y, t_2^{n+1}\}$ \\
      \hline
      $\CloneS_0$ & $\{f \mid f \text{ is $0$-separating}\}$ & $\{x \to y\}$ \\
      \hline
      $\CloneS_1$ & $\{f \mid f \text{ is $1$-separating}\}$ & $\{x \wedge \neg y\}$ \\
      \hline
      $\CloneS^n_{00}$ & $\CloneS^n_0 \cap \CloneR_2 \cap \CloneM$ & $\{x \lor  (y  \land  z), t_2^{n+1}\}$ \\
      \hline
      $\CloneS_{00}$ & $\CloneS_0 \cap \CloneR_2 \cap \CloneM$ & $\{x \lor  (y  \land  z)\}$ \\
      \hline
      $\CloneS_{12}$ & $\CloneS_1 \cap \CloneR_2$ & $\{x \land  (y  \to  z)\}$ \\
      \hline
      $\CloneS_{10}$ & $\CloneS_1 \cap \CloneR_2 \cap \CloneM$ & $\{x \land  (y  \lor  z)\}$ \\
      \hline

      $\CloneD$ & $\{f \mid f \text{ is self-dual}\}$ & $ \{(x \land \neg y) \lor (x \land \neg z) \lor (\neg y \land  \neg z)\}$ \\
      \hline
      $\CloneD_1$ & $\CloneD \cap \CloneR_2$ & $\{d_1\}$ \\ 
      \hline
      $\CloneD_2$ & $\CloneD \cap \CloneM$ & $\{t_2^3\}$ \\
      \hline

      $\CloneL$ & $\{f \mid f \text{ is affine}\}$ & $\{ x \xor y,\true\}$ \\
      \hline

      $\CloneV$ & $\{f \mid  f $ is a disjunction of variables or constants$\}$ & $\{ x \lor y, \false,\true \}$ \\
      \hline

      $\CloneE$ & $\{f \mid  f $ is a conjunction of variables or constants$\}$ & $\{ x \land y, \false, \true \}$ \\
      \hline


      \specialrule{\heavyrulewidth}{0cm}{0cm}
    \end{tabular}
    \smallskip
    \caption{\label{tab:clones}%
      List of relevant Boolean clones with definitions and bases, where $t^q_p$ denotes the $q$-ary $p$-threshold function and
      $d_1(x,y,z) = (x \land y) \lor (x \land \neg z) \lor (y \land \neg z)$.}
  \end{table*}

  \begin{figure*}[t]
    \centering
    \includegraphics[width = 10.5cm]{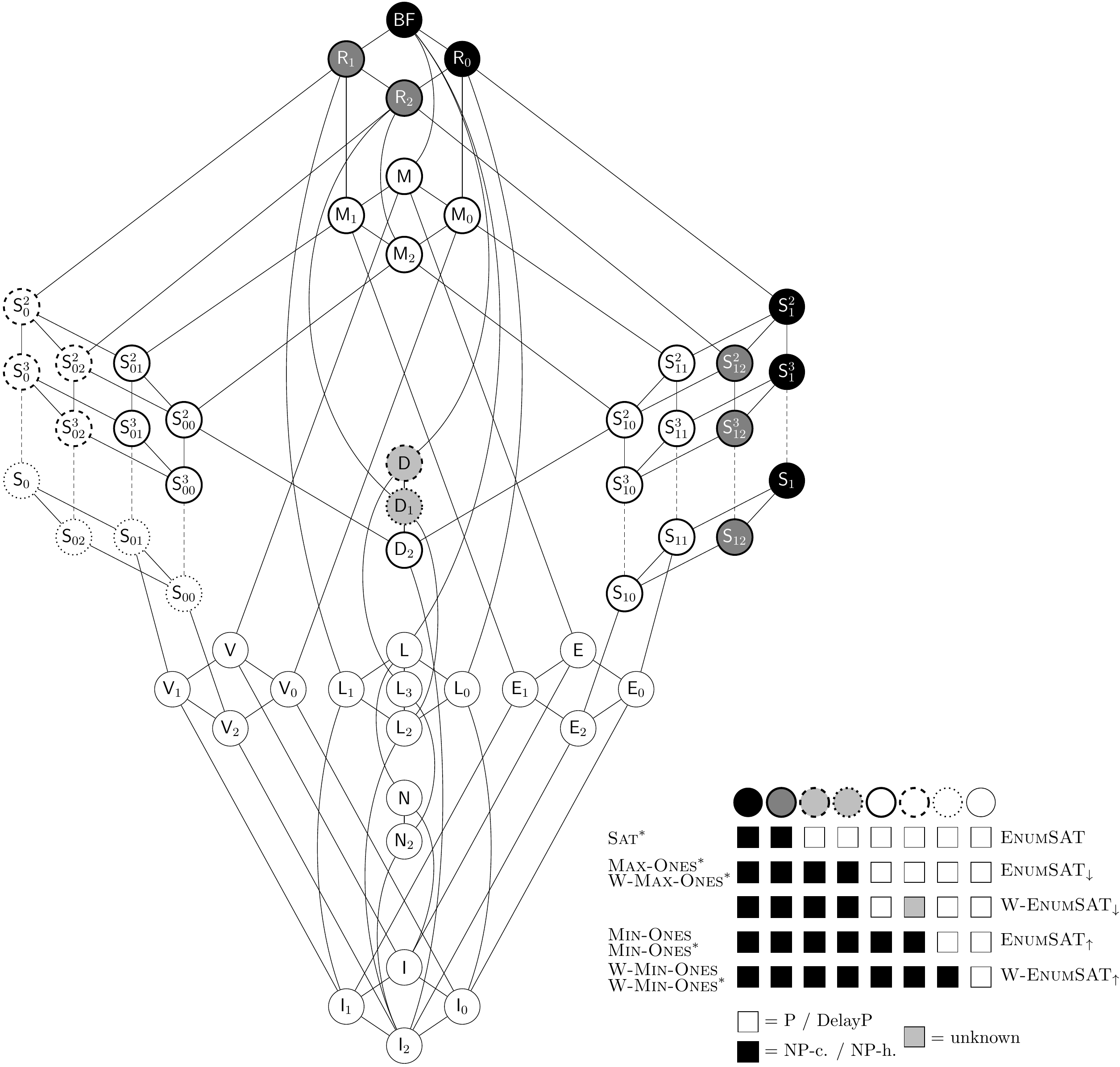}
    \caption{\label{fig:complexity} The complexity of all problems from this paper illustrated on Post's Lattice.}
  \end{figure*}

\section{Enumeration without Order Prescription}

We begin by looking at the model enumeration problem without order prescription.
\problemdef
{$\EnumSAT(B)$}
{a $B$-formula $\varphi$}
{generate all models of $\varphi$ (without duplicates)}
This problem has been studied in \cite{BoCrGaReScVo12} considering $B$-circuits instead of $B$-formul{\ae}.
It is not difficult to observe that the algorithms from \cite{BoCrGaReScVo12} also prove $\SpaceDelayP$-membership for $B$-formul{\ae} for the clones $\CloneM$, $\CloneL$, $\CloneD$ and $\CloneS_0^2$.
But we have to slightly modify the hardness proof from \cite{BoCrGaReScVo12} in order to deal with the issue of possible exponential blowup.
Hardness of $\EnumSAT(B)$ is inherited from $\SAT^*(B)$, the non-trivial satisfiability problem for $B$-formul{\ae} (given a $B$-formula, does it admit a non-trivial model $m$, \ie, $m \notin \{\vec{0},\vec{1}\}$?).
A look at Post's lattice shows us that $\CloneS_{12} \not\subseteq [B]$ if and only if either $[B] \subseteq \CloneM$, or $[B] \subseteq \CloneL$, or $[B] \subseteq \CloneD$, or $[B] \subseteq \CloneS_0^2$. The following proposition will therefore complete the classification.

	\begin{proposition}\label{prop:SATstarNPhard}
		Let $\CloneS_{12} \subseteq [B]$. Then $\SAT^*(B)$ is $\NP$-complete.
  \end{proposition}

  \begin{theorem}
    Let $B$ be a finite set of Boolean functions. Then $\EnumSAT(B)$  is
    \begin{enumerate}
      \item $\NP$-hard if $\CloneS_{12}\subseteq [B]$,
      \item in $\SpaceDelayP$ otherwise (\ie, $[B] \subseteq \CloneM$ or $[B] \subseteq \CloneL$ or $[B] \subseteq \CloneD$ or $[B] \subseteq \CloneS_0^2$).
    \end{enumerate}
  \end{theorem}

\section{Enumeration by Order of Non-decreasing Weight}

In this section we consider model enumeration by order of non-decreasing weight.
\problemdef
{$\EnumSATinc(B)$}
{a $B$-formula $\varphi$}
{generate all models of $\varphi$ by order of non-decreasing weight}
	\begin{proposition}\label{prop:EnumSATinc-VELS0-inDelayP}
		Let $[B] \subseteq \CloneV$ or $[B] \subseteq \CloneE$ or $[B] \subseteq \CloneL$ or $[B] \subseteq \CloneS_0$. Then $\EnumSATinc(B) \in \SpaceDelayP$.
		\rem{proof = APPENDIX}
  \end{proposition}
\begin{proof}
The first three cases are easy. More interesting is the fourth case. Let $\varphi$ be a $B$-formula with $n$ variables. Since we are $0$-separating, we know that there is a special variable, call it $x_j$, such that any assignment with $x_j = 1$ is a model. The number of assignments of weight $k$ with $x_j=1$ (which all are satisfying assignments, we call them therefore \emph{steady} models) is $\binom{n-1}{k-1}$, while the number of assignments of weight $k$ with $x_j=0$ is $\binom{n-1}{k}$. Since the factor between $\binom{n-1}{k-1}$ and $\binom{n-1}{k}$ is polynomial, the output process of the \emph{steady} models delivers enough time to determine in the meanwhile
the set of \emph{unsteady} models, that is, models of weight $k$ with $x_j=0$. These can be stored and output afterwards.
Note that this method uses exponential space. One can however obtain polynomial space (still maintaining polynomial delay) by not storing for each $k$ the whole set of \emph{unsteady} models, but by starting outputting them while still outputting the \emph{steady} ones.
\end{proof}

Solving $\EnumSATinc$ requires to efficiently solve $\MINONES$, the task of computing a model of minimal weight.
We will therefore inherit hardness from $\MINONES$.

	\begin{proposition}\label{prop:minones-np-hard}
		Let $\CloneS_{10} \subseteq [B]$ or $\CloneS_{00}^n \subseteq [B]$ for an $n \geq 2$ or $\CloneD_2 \subseteq [B]$.
		Then $\MINONES(B)$ is $\NP$-hard.
  \end{proposition}
	\begin{proof}
	In the all three cases we reduce from $\MINONES$(positive-2CNF) ($\NP$-hard according to \cite{KhSuWi97}). The second and third case are technically involving, where we deal with the $q$-ary $p$-threshold function.
	\end{proof}

  \begin{theorem}\label{thm:minones_classification}
Let $B$ be a finite set of Boolean functions. Then $\EnumSATinc(B)$ 
\begin{enumerate}
\item is $\NP$-hard if $\CloneS_{00}^n\subseteq [B]$ for some $n \geq 2$ or $\CloneD_2 \subseteq [B]$ or $\CloneS_{10}\subseteq [B]$,
\item is in $\SpaceDelayP$ otherwise (\ie, $[B] \subseteq \CloneS_0$ or $[B] \subseteq \CloneV$ or $[B] \subseteq \CloneL$ or $[B] \subseteq \CloneE$).
\end{enumerate}
	\end{theorem}



\noindent
We turn to the weighted variant.

\noindent
The following method will deliver us several tractability results.
\begin{theorem}[Priority queue method \cite{JoPaYa88,Schmidt09}]\label{thm:delaypmethod}
	Let $E=(I,Sol,\leq)$ be an enumeration problem. If it holds
	\begin{enumerate}
		\item \label{enum:totpoly}				for each $x\in I$, $\leq$ restricted to $Sol(x)$ is total and computable in polynomial time in $|x|$,
		\item \label{enum:minsolpoly}			it can be determined in polynomial time in $|x|$ whether $Sol(x)$ is non-empty and if so,
																			then $\min(Sol(x))$ is computable in polynomial time in $|x|$,
		\item \label{enum:f}							there is a binary function $f$ such that for all $x\in I$ and for all $y\in Sol(x)$ holds:
			\begin{enumerate}
				\item \label{enum:fpoly}				$f(x,y)$ is computable in polynomial time in $|x|$		
				\item \label{enum:fsol}					$f(x,y) \subseteq Sol(x)$ 														
				\item \label{enum:fcov}					if $y\neq \min(Sol(x))$ then there is a $z\in Sol(x)$ such that $z < y$ and $y\in f(x,z)$,
			\end{enumerate}

	\end{enumerate}
	then $E\in\DelayP$.
\end{theorem}

\begin{proof}

Correctness of the following algorithm is not difficult to observe.
	\begin{algorithmic}[1]
\newcommand{\LINEIF}[2]{\STATE\algorithmicif\ {#1}\ \algorithmicthen\ {#2}}
\newcommand{\LINEELSIF}[2]{\STATE\algorithmicelsif\ {#1}\ \algorithmicthen\ {#2}}
\newcommand{\LINEELSE}[1]{\STATE\algorithmicelse\ {#1}}
\newcommand{\LINERETURN}{\algorithmicreturn\ }
\renewcommand{\algorithmiccomment}[1]{\STATE//\;#1}
			\LINEIF{$Sol(x)=\emptyset$}{\LINERETURN 'no'}
			\STATE Q = newPriorityQueue($\leq$)
			\STATE compute $\ell := \min(Sol(x))$
			\STATE Q.enqueue($\ell$)

			\WHILE{Q is not empty}
				\STATE $\ell$ := Q.dequeue
				\STATE output $\ell$
				\STATE compute $L:=f(x,\ell)$
				\FORALL{$z\in L$}
					\LINEIF{$z>\ell$}{Q.enqueue($z$)}
				\ENDFOR
			\ENDWHILE
		\end{algorithmic}
The priority queue is supposed to eliminate duplicates.
Note that this method may run in exponential space.
\end{proof}
In order to apply the method to the partial order induced by the weight of assignments, it suffices to extend it to a total order, for instance by the lexicographical order on assignments.

	\begin{proposition}
		Let $[B] \subseteq \CloneV$ or $[B] \subseteq \CloneE$. Then $\wEnumSATinc(B) \in \DelayP$.
  \end{proposition}
  \begin{proof}
In the first case a $B$-formula can be seen as disjunction of variables and constants.
All assignments are models, with the possible exception of $\vec{0}$.
Thus, we reduce our problem to 
\problemdef{$\subsetsum$}
{A sequence of non-negative integers $C = (w_1,\dots,w_n)\in\N^n$}%
{generate all subsets $S \subseteq \{1, \dots, n\}$ by non-decreasing weight $\delta(S)$, where $\delta(S)=\sum_{i\in S} w_i$}
This task can be solved in polynomial delay and polynomial space by a dynamic programming method if the weights on the variables are polynomially bounded \cite{CrOlSc11}. Otherwise, the priority queue method from Theorem~\ref{thm:delaypmethod} is applicable with
$f(C,S) = \{S \cup \{i\} \mid i \in \{1, \dots, n\}\}$.

In the second case a $B$-formula can be seen as conjunction of variables and constants.
If this disjunction contains a constant $C_0$, then there are no models.
Otherwise $\vec{1}$ is the only model, up to fictive variables occurring in constants $C_1$.
Again, we reduce our problem to $\subsetsum$ as in the previous case.
  \end{proof}

	\begin{proposition}\label{prop:wEnumSATinc-L-inDelayP}
		Let $[B] \subseteq \CloneL$. Then $\wEnumSATinc(B) \in \DelayP$.
  \end{proposition}
  \begin{proof}
		Apply Theorem~\ref{thm:delaypmethod} with 
			$f(\varphi,m) = \{m \cup \{x\} \mid x \text{ fictive}\} \; \cup \; \{m \cup \{x,y\} \mid x,y \text{ not fictive} \text{ and } m \cap \{x,y\} = \emptyset\}$.
  \end{proof}

\noindent
The following previously tractable fragment becomes intractable.
	\begin{proposition}\label{prop:wMINONES-S00-NPhard}
		Let $\CloneS_{00} \subseteq [B]$. Then $\wMINONES(B)$ is $\NP$-hard. \rem{proof = APPENDIX}
  \end{proposition}
  \begin{proof}
Via a reduction from $\MINONES(B \cup \{\false\})$, replacing $\false$ by a fresh variable of big weight.
  \end{proof}

  \begin{theorem}
    Let $B$ be a finite set of Boolean functions. Then $\wEnumSATinc(B)$
    \begin{enumerate}
      \item is $\NP$-hard if $\CloneS_{00}\subseteq [B]$ or $\CloneD_2 \subseteq [B]$ or $\CloneS_{10}\subseteq [B]$,
      \item is in $\DelayP$ otherwise (\ie, $[B] \subseteq \CloneV$ or $[B] \subseteq \CloneL$ or $[B] \subseteq \CloneE$).
    \end{enumerate}
  \end{theorem}

\section{Enumeration by Order of Non-increasing Weight}

In this section we consider model enumeration by order of non-increasing weight.

\problemdef
{$\EnumSATdec(B)$}
{a $B$-formula $\varphi$}
{generate all models of $\varphi$ by order of non-increasing weight}
Analogously to Proposition~\ref{prop:EnumSATinc-VELS0-inDelayP}
we obtain $\SpaceDelayP$-membership for disjunctive, conjunctive, affine, or 0-separating formul{\ae}.
	\begin{proposition}\label{prop:EnumSATdecc-VELS0-inDelayP}
		Let $[B] \subseteq \CloneV$ or $[B] \subseteq \CloneE$ or $[B] \subseteq \CloneL$ or $[B] \subseteq \CloneS_0$. Then $\EnumSATdec(B) \in \SpaceDelayP$.
  \end{proposition}
For monotone formul{\ae} in general we apply once more the priority queue method.
	\begin{proposition}\label{prop:EnumSATdec-M-inDelayP}
		Let $[B] \subseteq \CloneM$. Then $\EnumSATdec(B) \in \DelayP$.
  \end{proposition}
  \begin{proof}
		Apply Theorem~\ref{thm:delaypmethod} with $f(\varphi,m) = \{m \backslash \{x\} \mid m \backslash \{x\} \models \varphi\}$.
  \end{proof}

We now address one of the rare cases where the priority queue method is not applicable and still we obtain tractability. We use for this the following classical result from combinatorics.
	\begin{theorem}[Erd\H{o}s-Ko-Rado Theorem \cite{ErKoRa61}]\label{lem:ErKoRaTh}
		Let $n \geq 2r$ and $A$ be a family of distinct subsets of $\{1,\dots, n\}$ such that each subset is of size $r$ and each pair of subsets intersects. Then it holds
		$$|A| \leq \binom{n-1}{r-1}.$$
  \end{theorem}

	\begin{lemma}\label{lem:crucial}
		Let $f\in \CloneS_0^2$ be an $n$-ary Boolean function and let $k$ be an integer such that $n/2 \leq k \leq n$.
		Then the number of models of weight $k$ is at least $\binom{n-1}{k-1}$.
  \end{lemma}

\begin{proposition}\label{prop:EnumSATdec-S02-inDelayP}
	Let $[B] \subseteq \CloneS_0^2$. Then $\EnumSATdec(B) \in \DelayP$.
\end{proposition}

\begin{proof}
In a first step we give a description of the enumeration scheme for the weight range $n$ down to $n/2$.

\noindent
We start with weight $n$: there is one such assignment which is also a model (all functions in $\CloneS_0^2$ are 1-reproducing).
We continue with an inductive argument (for $n/2 \leq k < n$): assume that we know for weight $k$ exactly the set of models $S_k$.
By Lemma~\ref{lem:crucial}, we have $\binom{n-1}{k-1} \leq |S_k| \leq \binom{n}{k}$.
The total time needed to output these models is something polynomial in $|S_k|$. This delivers enough time to bruteforce all assignments of the next weight level $k-1$: There are $\binom{n}{k-1}$ such assignments to be tested, and the factor between $\binom{n}{k-1}$ and $|S_k|$ is obviously polynomially bounded in $n$. Summed up, while outputting (with polynomial delay) the models of weight $k$, we can compute the set of models of weight $k-1$. Repeated application of this allows to enumerate with polynomial delay all models in the weight range $n$ down to $n/2$ by order of non-increasing weight.

The models in the weight range $n/2$ down to $0$ can be computed and stored during the first step: When during the first step an assignment $a$ is tested, also test its complement, $\overline{a}$, which lies then in the weight range $n/2$ down to $0$. If $\overline{a}$ is a model, put it on a stack. After step 1 has finished, output all the assignments from the stack.

\end{proof}

\noindent
We turn to the intractable cases.
Solving $\EnumSATdec$ requires to efficiently solve $\MAXONES^*$, the task of computing a model of maximal weight different from $\vec{1}$.
The hardness of this task will therefore deliver us hardness of $\EnumSATdec$.
The hardness of $\MAXONES^*$ is obtained from $\SAT^*$ and the following problem.

\problemdef
{$\invRootWSAT$}
{a 3CNF-formula $\varphi$ of $n$ variables}
{does $\varphi$ admit a model of weight $\geq n - \sqrt{n}$?}

\begin{lemma}\label{lem:InvRootWeightSAT}
$\invRootWSAT$ is $\NP$-complete. It remains $\NP$-complete if the number of variables is assumed to be a power of $3$. \rem{proof = APPENDIX}
\end{lemma}

\begin{proposition}\label{prop:maxonesstar-np-hard}
	Let $\CloneS_{12} \subseteq [B]$ or $\CloneD_1 \subseteq [B]$. Then $\MAXONES^*(B)$ is $\NP$-hard.
\end{proposition}
\begin{proof}
	In the first case we reduce from $\SAT^*(B)$ via $\varphi \mapsto (\varphi,1)$ and conclude with Proposition~\ref{prop:SATstarNPhard}. In the second case we have a technically involving reduction from $\invRootWSAT$.
\end{proof}

  \begin{theorem}
    Let $B$ be a finite set of Boolean functions. Then $\EnumSATdec(B)$
    \begin{enumerate}
      \item is $\NP$-hard if $\CloneS_{12}\subseteq [B]$ or $\CloneD_1 \subseteq [B]$,
      \item is in $\DelayP$ otherwise (\ie, $[B] \subseteq \CloneS_0^2$ or $[B] \subseteq \CloneM$ or $[B] \subseteq \CloneL$), where $\EnumSATdec(X) \in \SpaceDelayP$ for $X \in \{\CloneV,\CloneE, \CloneL, \CloneS_0\}$
    \end{enumerate}
  \end{theorem}

\noindent
Lastly, a look at the weighted variant, where we obtain only partial results.

\begin{proposition}\label{prop:wEnumSATdec-S0-inDelayP}
	Let $[B] \subseteq \CloneS_0$ or $[B] \subseteq \CloneM$. Then $\wEnumSATdec(B) \in \DelayP$.
\end{proposition}
\begin{proof}
		Apply Theorem~\ref{thm:delaypmethod} with $f(\varphi,m) = \{m \backslash \{x\} \mid m \backslash \{x\} \models \varphi\}$.
\end{proof}

	\begin{proposition}\label{prop:wEnumSATdec-L-inDelayP}
		Let $[B] \subseteq \CloneL$. Then $\wEnumSATdec(B) \in \DelayP$.
  \end{proposition}
  \begin{proof}
  Analogously to Proposition~\ref{prop:wEnumSATinc-L-inDelayP}.
  \end{proof}
The following tractability indicates that also $\wEnumSATdec(\CloneS_0^2)$ might be tractable. However, none of the above algorithmic strategies seems to work out.
\begin{proposition}\label{prop:wMAXONESstar-S02-inP}
	Let $[B] \subseteq \CloneS_0^2$. Then $\wMAXONES^*(B) \in \P$.
\end{proposition}

\section{Conclusion}
In this paper we provided complete complexity classifications of the problem
of enumerating all satisfying assignments of a propositional $B$-formula for every set $B$ of allowed connectives, imposing the orders of non-decreasing weight and non-increasing weight. We also considered the weighted variant, where the variables are assigned a non-negative integer weight. We obtained a complete classification for the weighted variant when imposing the order of non-decreasing weight and remained with one open case for the order of non-increasing weight when the connectives are $0$-separating of degree $2$. Interesting are the polynomial delay algorithms we obtained. They either relay on combinatorial bounds allowing a brute force approach, or on the use of a priority queue which necessarily leads to an exponential space usage. Future research could affront the open case, but should also investigate the question of exponential space: can it be avoided, or is it inherent to these problems, in particular to $\subsetsum$ without polynomial bounds on the weights?

\medskip
\noindent
\textbf{Acknowledgements.} The author would like to thank Johan Thapper for combinatorial support.


\newpage

\section{Appendix}

\begin{tabular}{lll}
Proof of Proposition~\ref{prop:SATstarNPhard}\\

Proof of Proposition~\ref{prop:EnumSATinc-VELS0-inDelayP}\\

Proof of Proposition~\ref{prop:minones-np-hard}\\

Proof of Proposition~\ref{prop:wMINONES-S00-NPhard}\\

Proof of Lemma~\ref{lem:crucial}\\

Proof of Lemma~\ref{lem:InvRootWeightSAT}\\

Proof of Proposition~\ref{prop:maxonesstar-np-hard}\\

Proof of Proposition~\ref{prop:wMAXONESstar-S02-inP}\\

\end{tabular}

\bigskip

\noindent
Proof of Proposition~\ref{prop:SATstarNPhard} ($\SAT^*(B)$ is $\NP$-complete if $\CloneS_{12} \subseteq [B]$)
\begin{proof}
  $\NP$-membership is obvious. For the hardness, we give a reduction from the satisfiability problem for $B$-formul{\ae}. We know from \cite{Lewis79} that $\SAT(B')$ is $\NP$-complete if $\CloneS_1 \subseteq [B']$. Since $\CloneS_1 = [\CloneS_{12} \cup \{\false\}] \subseteq [B \cup \{\false\}]$, we conclude that $\SAT(B \cup \{\false\})$ is $\NP$-complete. Let $\varphi$ be a $B\cup \{\false\}$-formula with variable set $x_1, \dots, x_n$. We construct $\varphi' = \varphi[\false / f] \land \bigwedge_{i=1}^n t \land (f \rightarrow x_i)$. It is not difficult to verify that $\varphi$ is satisfiable if and only if $\varphi'$ admits a non-trivial model. Note that $\land$ and $x \land (y \imp z)$ are in $\CloneS_{12}$ and have therefore a $B$-representation. We transform $\varphi'$ into the final $B$-formula by replacing the connectives $\land$ and $x \land (y \imp z)$ with their $B$-representations. We avoid exponential blowup by a compact $\varphi'$: write the $n$-ary conjunction as a balanced tree of the binary conjunction $\land$.
\end{proof}

%

\bigskip

\noindent
Proof of Proposition~\ref{prop:EnumSATinc-VELS0-inDelayP}
($\EnumSATinc(X) \in \SpaceDelayP$ for $X \in \{\CloneV$, $\CloneE$, $\CloneL$, $\CloneS_0$\})
  \begin{proof}
In the first case a $B$-formula can be seen as disjunction of variables and constants.
All assignments are models, with the possible exception of $\vec{0}$.
We obviously can enumerate those assignments by order of non-decreasing weight by standard combinatorial methods.

In the second case a $B$-formula can be seen as conjunction of variables and constants.
If this disjunction contains a constant $C_0$, then there are no models.
Otherwise $\vec{1}$ is the only model, up to fictive variables occurring in constants $C_1$.
Again we can enumerate those assignments by order of non-decreasing weight by standard combinatorial methods.

In the third case a $B$-formula can be seen as linear equation over GF(2).
Therefore, the set of models is either the set of assignments with an even number of non-fictive variables set to $1$, or the set with an odd number of non-fictive variables set to $1$. Again, all these models can be enumerated by non-decreasing weight by standard combinatorial methods.

The fourth case is treated in the paper.
	\end{proof}

\bigskip

\noindent
Proof of Lemma~\ref{lem:crucial} (Let $f\in \CloneS_0^2$ be an $n$-ary Boolean function and let $k$ be an integer such that $n/2 \leq k \leq n$. Then the number of models of weight $k$ is at least $\binom{n-1}{k-1}$.)

\begin{proof}
Since the functions of $\CloneS_0^2$ are 0-separating of degree 2, the statement of this lemma is nothing else than a disguised form of the Erd\H{o}s-Ko-Rado Theorem \cite{ErKoRa61}: Set $r = n-k$ and represent an assignment as subset of $\{1, \dots, n\}$ by the set of indexes of the coordinates which are set to $0$. Then $A$ corresponds to the set of non-models.
\end{proof}

\bigskip

\noindent
Preparations for Proof of Proposition~\ref{prop:minones-np-hard} ($\MINONES(B)$ is $\NP$-hard if
$\CloneS_{10} \subseteq [B]$ or $\CloneS_{00}^n \subseteq [B]$ for an $n \geq 2$ or $\CloneD_2 \subseteq [B]$)

  \begin{lemma}\label{lem:minones_constant_1}
  	Let $\CloneE_2 \subseteq [B]$. Then $\MINONES(B\cup \{\true\}) \leqlogm \MINONES(B)$.
  \end{lemma}
  \begin{proof}
  	We map $\varphi$ to $\varphi' = \varphi[\true / t] \land t$ and $k$
  	to $k' = k + 1$. The formula $\varphi'$ can be written as a $B$-formula by
  	replacing the connective $\land$ by its $B$-representation
  	($\land \in \CloneE_2 \subseteq [B]$).
  \end{proof}

\begin{definition}
		We denote by $t_q^p$ the $p$-ary $q$-threshold function, $p > q \geq 2$.
		Define $^d\psi_q^p$ to be a complete $t_q^p$-tree of depth $d$.
\end{definition}
	\begin{lemma}\label{lem:minones-threshold}
		The formula $^d\psi_q^p$ \rem{proof = APPENDIX}
		\begin{enumerate}
			\item has arity $p^d$,
			\item evaluates to $0$ whenever less than $q^d$ inputs are set to $1$, and
			\item evaluates to $1$ whenever more than $p^d - q^d$ inputs are set to $1$.
		\end{enumerate}
  \end{lemma}
\begin{proof}
Obviously $^d\psi_q^p$ has arity $p^d$.

We prove the second statement by induction over $d$. For $d = 1$ we have $^1\psi_q^p = t_q^p$ and the statement holds by definition of $t_q^p$. Consider then $^{(d+1)}\psi_q^p = t_q^p(_1^d\psi_q^p, \dots ,_p^d\psi_q^p$). Assuming that none of the $_i^d\psi_q^p$ can be triggered with less than $q^d$ inputs set to 1, we observe that $^{(d+1)}\psi_q^p$ cannot be triggered with less than $q \cdot q^d = q^{d+1}$ inputs set to 1.

The third statement follows from the second by the self-duality of $^d\psi_q^p$.
\end{proof}

\begin{lemma}\label{lem:minones_constant_0}
  	If $t_q^p \in [B]$ for some $p,q \in \N$ with $p > q \geq 2$, then $\MINONES(B\cup \{\false\}) \leqlogm \MINONES(B)$. \rem{proof = APPENDIX, give here maybe a sketch}
  \end{lemma}

  \begin{proof}
  	Let $(\varphi, k)$ be an instance of $\MINONES(B\cup \{\false\})$,
  	$n$ the number of variables in $\varphi$.
		Choose $d\in\N$ such that $q^{d-1} < n < q^d$ (\ie, $d-1 < \log_q(n) < d$)
		and we obtain with Lemma~\ref{lem:minones-threshold} that $^d\psi_q^p$
		has arity and size polynomial in $n$
		and evaluates to $0$ whenever less than $n+1$ inputs are set to $1$
		($n + 1 \leq q^d$).

		Denote by $^dF_q^p$ the $B$-formula obtained from $^d\psi_q^p$ by replacing each
		$t^p_q$ by its $B$-representation. Exponential blowup does not occur since
		$^d\psi_q^p$ is compact.
  	We finally map $(\varphi, k)$ to $(\varphi', k')$, where
  	$\varphi' = \varphi[\false / ^dF_q^p(y_1, \dots, y_{p^d})]$
  	and $k' = \min(n,k)$
  	and the $y_1, \dots, y_{p^d}$ are fresh variables.
  	One easily verifies that $\varphi$ admits a model of weight $\leq k$ if and only if $\varphi'$ admits a model of weight $\leq \min(n,k)$.
  \end{proof}

\bigskip

\noindent
Proof of Proposition~\ref{prop:minones-np-hard} ($\MINONES(B)$ is $\NP$-hard if
$\CloneS_{10} \subseteq [B]$ or $\CloneS_{00}^n \subseteq [B]$ for an $n \geq 2$ or $\CloneD_2 \subseteq [B]$)
  \begin{proof}
  	In the first case we have $\CloneE_2 \subseteq \CloneS_{10} \subseteq [B]$, so we obtain by Lemma~\ref{lem:minones_constant_1} that $\MINONES(B\cup \{\true\}) \leqlogm \MINONES(B)$.
Since $\{\land, \lor\} \subseteq \CloneM_1 = [\CloneS_{10} \cup \{\true\}] \subseteq [B\cup \{\true\}]$, we can reduce from $\MINONES$(positive-2CNF) ($\NP$-hard according to \cite{KhSuWi97}) to $\MINONES(B\cup \{\true\})$ by replacing every connective ($\land, \lor$) by its ($B\cup \{\true\}$)-representation, avoiding exponential blowup by a compact representation of the multi-ary conjunction in the 2CNF-formula.

  	In the second case we have that $t_2^{n+1} \in [B]$, in the third case we have that $t_2^3 \in [B]$.
  	In both cases we obtain by Lemma~\ref{lem:minones_constant_0} that $\MINONES(B\cup \{\false\}) \leqlogm \MINONES(B)$.
  	Since further in both cases it holds $\{\land, \lor\} \subseteq \CloneM_2 \subseteq [B\cup \{\false\}]$, we can 
  	reduce from $\MINONES$(positive-2CNF) to $\MINONES(B\cup \{\false\})$ by replacing every connective ($\land, \lor$) by its ($B\cup \{\false\}$)-representation
  	(avoiding exponential blowup by a compact formula).
  \end{proof}

\noindent
Proof of Proposition~\ref{prop:wMINONES-S00-NPhard} ($\wMINONES(B)$ is $\NP$-hard)
  \begin{proof}
We give a reduction from $\MINONES(B \cup \{\false\})$, an $\NP$-hard problem according to Theorem~\ref{thm:minones_classification}, since $\CloneD_2 \subseteq \CloneM_0 = [\CloneS_{00} \cup \{\false\}] \subseteq [B \cup \{\false\}]$.
Let $(\varphi, k)$ be an instance of $\MINONES(B\cup \{\false\})$, $n$ the number of variables in $\varphi$.
We map $(\varphi, k)$ to $(\varphi', k')$, where $\varphi' = \varphi[\false / f]$, $f$ a fresh variable and $k' = \min(n,k)$.
We set the weight of all variables from $\varphi$ to $1$ and the weight of $f$ to $n+1$.
One easily verifies that $\varphi$ admits a model of weight $\leq k$ if and only if $\varphi'$ admits a model of weight $\leq \min(n,k)$.
  \end{proof}

\bigskip

\noindent
Proof of Lemma~\ref{lem:InvRootWeightSAT} ($\invRootWSAT$ is $\NP$-complete)
\begin{proof}
The $\NP$-membership is obvious: guess an assignment of weight $\geq n - \sqrt{n}$ and verify whether it is a satisfying one.

By flipping all literals in an instance of $\invRootWSAT$, we obtain equivalence to 

\problemdef
{$\RootWSAT$}
{a 3CNF-formula $\varphi$ of $n$ variables}
{does $\varphi$ admit a model of weight $\leq \sqrt{n}$?}

We give a reduction from $\MINONES$(3CNF) ($\NP$-hard according to \cite{KhSuWi97}) to $\RootWSAT$.
Let $(\varphi, k)$ be an instance of $\MINONES$(3CNF), let $n$ be the number of variables of $\varphi$.
\begin{itemize}

\item
If we have $k \geq \sqrt{n}$, we set $\ell = \min(k,n)$ and $r = \ell^2 - n$ and map the instance $(\varphi,k)$ to $\varphi' = \varphi \land \neg y_1 \land \dots \land \neg y_r$.
\\\\
Let $(\varphi,k) \in \MINONES$. By definition of $\ell$ this implies that $(\varphi,\ell) \in \MINONES$. That is, there is a model $\sigma$ of weight $w(\sigma) \leq \ell$. $\varphi'$ has $n' = n+r = \ell^2$ variables. The model $\sigma$ can be extended to a model $\sigma'$ of $\varphi'$ by setting additionally the variables $y_1, \dots, y_r$ to $0$. This $\sigma'$ has then weight 
$$w(\sigma') = w(\sigma) \leq \ell = \sqrt{n'}.$$
That is, $\varphi' \in \RootWSAT$.
\\\\
Conversely, let $\varphi' \in \RootWSAT$. That is, there is a model $\sigma'$ with
$$w(\sigma') \leq \sqrt{n'} =  \ell$$
Construct $\sigma$ from $\sigma'$ by removing the assignments for the $y_i$ and we get that $\sigma$ satisfies $\varphi$ and $w(\sigma) \leq \ell$.
That is, $(\varphi,\ell) \in \MINONES$ and, since $\ell \leq k$, $(\varphi,k) \in \MINONES$.

\item
Else if $k < \sqrt{n}$, we set $t = \sqrt{n-k+\frac{1}{4}}-k+\frac{1}{2}$.
One easily verifies that
\begin{equation}\label{equ:useful0}
k+t = \sqrt{n + t}\\
\end{equation}
and that consequently
\begin{equation}\label{equ:useful2}
\sqrt{n + \lfloor t\rfloor} \leq \sqrt{n+t} = k + t < k + \lfloor t\rfloor+1
\end{equation}
From the property that for $x\geq 0$ we have $\lfloor\sqrt{x}\rfloor \leq \sqrt{\lfloor x\rfloor}$ one obtains that
\begin{equation}\label{equ:useful1}
k + \lfloor t \rfloor = \lfloor k + t \rfloor =  \lfloor \sqrt{n + t} \rfloor \leq \sqrt{\lfloor n + t\rfloor} = \sqrt{n+\lfloor t\rfloor}
\end{equation}
\\\\
Set now $r = \left\lfloor t\right\rfloor$ and map the instance $(\varphi,k)$ to $\varphi' = \varphi \land y_1 \land \dots \land y_r$.
\\\\
Let $(\varphi,k) \in \MINONES$. That is, there is a model $\sigma$ of weight $w(\sigma) \leq k$. $\varphi'$ has $n' = n+r$ variables.
The model $\sigma$ can be extended to a model $\sigma'$ of $\varphi'$ by setting additionally the variables $y_1, \dots, y_r$ to $1$. This $\sigma'$ has then weight
$$\textstyle w(\sigma') = w(\sigma)+r \leq k + r 
\overset{(\ref{equ:useful1})}{\leq}
\sqrt{n+r} = \sqrt{n'}.$$
That is, $\varphi' \in \RootWSAT$.
\\\\
Conversely, let $\varphi' \in \RootWSAT$. That is, there is a model $\sigma'$ with $w(\sigma') \leq \sqrt{n'} = \sqrt{n+r} \overset{(\ref{equ:useful2})}{<} k + r + 1$.
That is, $w(\sigma') \leq k + r$.
By construction of $\varphi'$, we know that $\sigma'$ has to set all $y_1, \dots, y_r$ to $1$. Thus, by reducing $\sigma'$ by the values for the $y_i$, we obtain an assignment $\sigma$ with $w(\sigma) \leq k$. By construction of $\varphi'$, we conclude that $\sigma$ is indeed a model of $\varphi$. That is, $(\varphi,k) \in \MINONES$.

\end{itemize}

\bigskip

To prove the second statement, add enough dummy variables to fill up to the next power of $3$, be $d$ such that $3^{d-1} < n \leq 3^d$.
Force the right amount of them to $1$ and the rest to $0$ in order to map models of weight $\leq \sqrt{n}$ to models of weight $\leq \sqrt{3^d}$ and models of weight $> \sqrt{n}$ to models of weight $> \sqrt{3^d}$.
\end{proof}

\bigskip

\noindent
Proof of Proposition~\ref{prop:maxonesstar-np-hard} ($\MAXONES^*(B)$ is $\NP$-hard if $\CloneS_{12} \subseteq [B]$ or $\CloneD_1 \subseteq [B]$)
\begin{proof}
In the first case we reduce from $\SAT^*(B)$ via $\varphi \mapsto (\varphi,1)$ and conclude with Proposition~\ref{prop:SATstarNPhard}.

In the second case  we give a reduction from $\invRootWSAT$ with the assumption that the number of variables is a power of $3$.
	Let $\varphi$ be an instance of $\invRootWSAT$ with variable set $x_1, \dots, x_n$, where $n = 3^d$.
	We transform $\varphi$ in several steps. We assume $\varphi$ to be a compact $\{\land, \lor, \neg\}$-formula.	
	
	\begin{enumerate}
	
	\item Since $\{\land, \lor, \neg\} \subseteq \CloneBF = [\CloneR_1 \cup \{\false\}]$ and $\CloneR_1 = [\lor, =]$, we can transform $\varphi$ into a compact
	$\{\lor, =, \false\}$-formula $\varphi_1$ by replacing the connectives $\{\land, \lor, \neg\}$ with their $\{\lor, =, \false\}$-representations.
	
	\item Set $\varphi_2 = \varphi_1[\false/f] \land \bigwedge_{i=1}^{n}(f \to x_i)$, where $f$ is a fresh variable and the $n$-ary conjunction be represented in
				a compact way by the binary conjunction.
				Note that $\varphi_2$ is now a compact $\{\lor, =, \land, \to\}$-formula.
	
	\item Since $\{\lor, =, \land, \to\} \subseteq \CloneR_1 = [\CloneD_1 \cup \{\true\}]$ and $\CloneD_1 = [d_1]$, we can transform $\varphi_2$ into a compact $\{d_1, \true\}$-formula
				$\varphi_3$ by replacing the connectives $\{\lor, =, \land, \to\}$ with their $\{d_1, \true\}$-representations.
	
	\item Set $\varphi_4 = \varphi_3[\true/ ^d\psi_2^3(x_1, \dots x_{3^d})]$. Observe that $\varphi_4$ is a compact $\{d_1, t_2^3\}$-formula of size polynomial in the size of $\varphi$
	with variable set $x_1, \dots, x_n,f$.
	
	\item At last, since $d_1,t_2^3 \in D_1 \subseteq [B]$, we can transform $\varphi_4$ into $\varphi_5$ by replacing the connectives $d_1,t_2^3$ with their $B$-representations. Note that $\varphi_5$ is still of polynomial size, since $\varphi_4$ is compact.
	\end{enumerate}

	We finally map $\varphi$ to $(\varphi_5, \lceil n - \sqrt{n} \rceil)$.
	It is not difficult to verify that $\varphi$ admits a model of weight $w$ with $\lceil n-\sqrt{n} \rceil \leq w \leq n$ if and only if $\varphi_5$ admits a model of weight $v$ with
	$\lceil n-\sqrt{n} \rceil \leq v < n+1 = |\Vars{\varphi_5}|$ (to pass from one model to another, add/remove the assignment $f = 0$).

\end{proof}

\bigskip

\noindent
Proof of Proposition~\ref{prop:wMAXONESstar-S02-inP} ($\wMAXONES^*(\CloneS_0^2) \in \P$)
\begin{proof}
Let $n$ be the number of variables in an instance.
We obtain a model different from $\vec{1}$ of maximal weight by searching among the set of assignments with one $0$ and $(n-1)$ $1$'s.
The number of such assignments is obviously polynomial in $n$ and the property of being 0-separating of degree 2 guarantees us among them a non-empty set of models.
Obviously a maximum weight model different from $\vec{1}$ is among them (all assignments with more than one $0$ are of less or equal weight).
\end{proof}

\end{document}